\documentclass[preprint]{elsarticle}

\journal{Journal of Multivariate Analysis}

\usepackage{amsmath,amssymb,amsthm,mathtools,bm}
\usepackage{graphicx}
\usepackage{booktabs}
\usepackage{multirow}
\usepackage{array}
\usepackage{algorithm}
\usepackage{algpseudocode}
\usepackage{siunitx}
\usepackage{adjustbox}
\usepackage{threeparttable}
\usepackage{hyperref}
\hypersetup{colorlinks=true,linkcolor=blue,citecolor=blue,urlcolor=blue}
\newcommand{\E}{\mathbb{E}}

\newcommand{\R}{\mathbb{R}}
\newcommand{\I}{\mathbf{I}}
\newcommand{\tr}{\mathrm{tr}}

\newcommand{\KL}{\mathrm{KL}}
\newcommand{\TV}{\mathrm{TV}}
\newcommand{\op}{\mathrm{op}}
\newcommand{\Fhyper}{{}_2F_1}

\theoremstyle{plain}
\newtheorem{theorem}{Theorem}
\newtheorem{lemma}{Lemma}
\newtheorem{proposition}{Proposition}

\theoremstyle{remark}
\newtheorem{remark}{Remark}
\theoremstyle{definition}
\newtheorem{definition}{Definition}

\begin{document}

\begin{frontmatter}

\title{HIMCE: High-dimensional multiple imputation via covariance-mode updating for neuroimaging and spatiotemporal blocks}

\author[ucf]{Hsin-Hsiung Huang\corref{cor1}}
\ead{Hsin.Huang@ucf.edu}
\author[tno,uu]{Stef van Buuren}
\ead{S.vanBuuren@uu.nl}
\cortext[cor1]{Corresponding author.}

\affiliation[ucf]{organization={School of Data, Mathematical, and Statistical Sciences, University of Central Florida}, city={Orlando}, state={FL}, country={USA}}
\affiliation[tno]{organization={Netherlands Organization for Applied Scientific Research (TNO)}, city={Leiden}, country={The Netherlands}}
\affiliation[uu]{organization={Department of Methodology \& Statistics, Faculty of Social Sciences (FSS), Utrecht University}, city={Utrecht}, country={The Netherlands}}

\begin{abstract}
High-dimensional neuroimaging and spatiotemporal blocks often contain structured missingness from acquisition artifacts, preprocessing failures, and sensor dropout. Multiple imputation propagates uncertainty, but fully conditional specification methods such as multivariate imputation by chained equations (MICE) can be slow or unstable when block dimension is large and correlations are strong. A multivariate normal (MVN) working model provides a coherent posterior predictive target and an exact data augmentation sampler, but repeated covariance sampling and matrix factorizations become costly in large dimensions.

We propose \textbf{H}igh-dimensional \textbf{I}mputation via covariance \textbf{M}ode and \textbf{C}hained \textbf{E}quations (\textbf{HIMCE}), a hybrid multiple-imputation procedure for continuous blocks. Relative to exact MVN data augmentation, HIMCE preserves the Gaussian conditional imputation law and propagates mean-parameter uncertainty through stochastic coefficient or local-ridge draws. In high-dimensional blocks, it approximates covariance uncertainty through covariance-mode updating, optionally with a scalar bridge; in small blocks, it can restore exact covariance uncertainty through a conditional inverse-Wishart refresh.

We record the exact Bayesian reference sampler and prove fixed-dimensional posterior consistency and asymptotic equivalence of mode plug-in prediction in total variation. We also develop diagnostics based on randomized rank-cell probability integral transform (PIT), PIT-consistent empirical coverage, and marginal distribution overlays. In the primary spatial benchmark, HIMCE improves posterior-mean error relative to HIMA and screened MICE, runs at HIMA-like speed and below half the MICE runtime, and improves interval coverage over HIMA, although MICE remains better calibrated. A repeated low-dimensional NHANES illustration shows improved coverage with competitive point prediction.
\end{abstract}


\end{frontmatter}

\section{Introduction}
Missing values are pervasive in neuroimaging and spatiotemporal studies. In functional MRI and diffusion pipelines, voxelwise or regionwise blocks may be partially missing due to motion artifacts, susceptibility distortions, registration failures, and quality-control exclusions. In longitudinal sensor arrays, intermittent dropout and preprocessing failures yield structured gaps across time and space. Analyses that discard incomplete records can be inefficient and can be biased when missingness depends on observed covariates and observed outcomes, while single imputation can substantially understate uncertainty. Multiple imputation (MI) addresses both issues by generating $M$ completed datasets, applying a complete-data analysis to each, and combining results in a way that reflects within-imputation and between-imputation uncertainty \citep{Rubin2004}.

Two broad families of MI procedures dominate practice. Fully conditional specification (FCS) approaches, including multivariate imputation by chained equations (MICE), iterate a collection of conditional models and are popular because they allow flexible modeling choices \citep{vanBuuren2011,vanBuuren2018}. In high-dimensional continuous blocks, however, MICE can become slow or unstable. Strong collinearity among block variables yields ill-conditioned regressions, and repeated estimation of high-dimensional conditional models can fail numerically or require aggressive predictor screening and regularization.

Joint-model MI under a multivariate normal (MVN) working model provides a coherent posterior predictive target and an exact data augmentation sampler \citep{Schafer1997}. The joint model is particularly appealing for continuous blocks because it respects cross-variable dependence and produces proper predictive uncertainty when the working model is calibrated. Yet the standard sampler relies on repeated inverse-Wishart updates of a $p\times p$ covariance matrix, along with repeated matrix factorizations for conditional Gaussian draws. When $p$ is large, or when $p$ is comparable to or larger than $n$, these steps can become computational bottlenecks and numerically delicate \citep{SchaferGraham2002}.

This paper develops HIMCE, a hybrid MI procedure for continuous blocks. HIMCE keeps the joint Gaussian block model and propagates uncertainty in the regression mean through stochastic coefficient or local-ridge draws. The high-dimensional branch replaces matrix-valued covariance sampling by a stabilized covariance-mode update. The resulting chain is not an exact joint-model sampler; it is a posterior-motivated stochastic approximation whose main computational saving is the replacement of repeated inverse-Wishart covariance draws by a mode plug-in. The practical implementation also includes a dimension-adaptive option: when the block is small and exact covariance sampling is inexpensive, the conditional inverse-Wishart refresh is restored.

A central theme is the preservation--approximation split induced by covariance-mode updating. Relative to exact MVN data augmentation, HIMCE preserves the Gaussian working model, the conditional Gaussian law for each missingness pattern, and stochastic propagation of mean-parameter uncertainty. In the high-dimensional branch it approximates covariance uncertainty through covariance-mode updating and, optionally, a one-dimensional scalar bridge around the mode. Thus predictive center and dominant conditional dependence are retained by the Gaussian structure, whereas full posterior uncertainty in the covariance matrix is deliberately compressed. This split motivates the paper's emphasis on pseudo-missing calibration diagnostics rather than point prediction alone.

The companion benchmark also highlights three practical points. First, in finite imputation ensembles the diagnostic layer matters: PIT histograms, empirical coverage, and marginal checks should all be computed from the same empirical predictive distribution. Second, finite-sample center and scale adjustments should be learned on observed cells, not tuned after inspecting withheld pseudo-missing truths. Third, comparisons with screened MICE should state the screening rule, number of retained predictors, number of iterations, number of imputations, and regularization choices, because those choices materially affect both stability and runtime.

The remainder is organized as follows. Section~\ref{sec:related} positions HIMCE relative to FCS and joint-model MI in high dimension. Section~\ref{sec:method} introduces the Gaussian block model, the conjugate reference posterior, HIMA as the deterministic covariance-mode baseline, and HIMCE as the stochastic covariance-mode chain. Section~\ref{sec:diagnostics} discusses multiple-imputation inference and practical distributional diagnostics. Section~\ref{sec:theory} records fixed-dimension theoretical results. Section~\ref{sec:sim} reports a spatially correlated pseudo-missing benchmark. Section~\ref{sec:nhanes} adds a low-dimensional NHANES illustration. Section~\ref{sec:practical} discusses practical implementation choices. The appendix records proofs and supporting remarks.

\section{Related work and positioning}
\label{sec:related}
Multiple imputation has a large literature; here we focus on ingredients most relevant for high-dimensional continuous blocks with strong dependence and structured missingness patterns.

\subsection{Multiple imputation and fully conditional specification}
Rubin's framework formalizes multiple imputation as Monte Carlo sampling from a posterior predictive distribution under an imputation model, followed by complete-data analysis and pooling to reflect within- and between-imputation uncertainty \citep{Rubin2004,LittleRubin2019}. A widely used practical strategy is fully conditional specification, in which one specifies a set of conditional models and iterates them to form a Markov chain whose stationary distribution is used as a working approximation for posterior predictive draws. MICE is the most common implementation of FCS in applied work \citep{vanBuuren2011,vanBuuren2018}. In moderate dimensions, MICE can be flexible and effective because each conditional model may be tailored to the variable type and scientific context.

High-dimensional blocks stress the standard MICE workflow. Each conditional regression can become ill conditioned when predictors are strongly collinear, and instability is amplified when $p$ is comparable to or larger than $n$. Repeatedly fitting $p$ conditional models can also be slow. In practice, high-dimensional MICE analyses therefore rely on predictor screening, regularization, or dimension reduction. These stabilizations can help, but they also alter the fitted conditional system and may perform poorly when the block dependence is dense.

\subsection{Joint-model MI for continuous data}
For continuous blocks, joint-model MI under a multivariate normal working model is a classical alternative to FCS \citep{Schafer1997}. The joint model induces coherent conditional distributions for any missingness pattern and admits an exact data augmentation sampler based on alternating conditional Gaussian draws for missing values and conjugate updates for regression and covariance parameters. In moderate dimensions, this yields stable multiple imputations with a clean posterior predictive interpretation. A recurring limitation in large blocks is that exact sampling requires repeated updates and factorizations of a $p\times p$ covariance matrix, which can be slow and numerically delicate when the covariance is ill conditioned or when $p>n$ \citep{SchaferGraham2002}.

\subsection{Covariance regularization and covariance-mode strategies}
High-dimensional covariance estimation is itself a stability problem. Shrinkage estimators regularize the sample covariance toward a well-conditioned target, such as a scaled identity, and can remain stable when $p>n$ \citep{LedoitWolf2004}. In the imputation context, a related idea is to avoid repeated inverse-Wishart covariance sampling and instead update the covariance by a stabilized point estimate or posterior mode that is guaranteed to be positive definite under mild regularization. Such covariance-mode strategies can substantially reduce runtime and numerical failures in large blocks, but they also move the algorithm away from exact posterior sampling, motivating calibration checks based on pseudo-missing experiments.

\subsection{High-dimensional neuroimaging imputation: HIMA and HIMCE}
Recent work has adapted covariance-mode ideas to neuroimaging blocks. The High-dimensional Imputation via covariance Mode Approximation (HIMA) approach \citep{lu2025new} is the deterministic baseline in our benchmark implementation. HIMA fits a Gaussian block model, but its covariance step is the empirical-Bayes approximate covariance mode described in the Human Brain Mapping paper. In the current companion benchmark we use that empirical-Bayes covariance-mode estimator as the HIMA covariance backbone.

HIMCE is complementary. It keeps the same Gaussian conditional imputation structure and, in the high-dimensional experiments, the same HIMA covariance backbone. It differs from HIMA by replacing deterministic mean updates with stochastic coefficient or screened local-ridge draws, and by allowing lightweight finite-sample calibration of stored draws using observed cells. The theoretical section does not attempt a fixed-dimension proof for the exact empirical-Bayes covariance estimator used by HIMA; instead it establishes asymptotic validity for a generic covariance-mode plug-in family that captures the same preservation--approximation principle.

\section{Model and methodology}
\label{sec:method}
\subsection{Setup, notation, and missingness}
Let $Y\in\R^{n\times p}$ denote a continuous block, such as a region of interest reshaped as subject by voxel, or a sensor by time block. Let $X\in\R^{n\times k}$ be a fixed design matrix of auxiliary covariates. For subject $i$, the row vector $Y_i\in\R^p$ contains observed and missing entries. Let $R\in\{0,1\}^{n\times p}$ denote missingness indicators with $R_{ij}=1$ if $Y_{ij}$ is observed.

We work under a missing-at-random assumption,
\begin{equation}
\mathbb{P}(R\mid Y,X)=\mathbb{P}(R\mid Y_{\mathrm{obs}},X).
\label{eq:mar}
\end{equation}
Under MAR and an imputation model $p(Y\mid X,\theta)$, MI targets posterior predictive draws
\begin{equation}
Y_{\mathrm{mis}}^{(m)} \sim p(Y_{\mathrm{mis}}\mid Y_{\mathrm{obs}},X)
= \int p(Y_{\mathrm{mis}}\mid Y_{\mathrm{obs}},X,\theta)\,p(\theta\mid Y_{\mathrm{obs}},X)\,d\theta.
\label{eq:pp_target}
\end{equation}
The goal is to approximate \eqref{eq:pp_target} in a way that remains stable when $p$ is large and the block is strongly correlated.

\subsection{Gaussian block model}
HIMCE adopts a Gaussian working model,
\begin{equation}
Y_i \mid X_i,B,\Sigma \sim \mathcal{N}_p(X_iB,\Sigma), \qquad i=1,\ldots,n,
\label{eq:model}
\end{equation}
where $B\in\R^{k\times p}$ is a regression coefficient matrix and $\Sigma\in\R^{p\times p}$ is a symmetric positive definite residual covariance. Equivalently,
\begin{equation}
Y \mid B,\Sigma \sim \mathrm{MN}_{n\times p}(XB,\I_n,\Sigma),
\label{eq:matrixnormal}
\end{equation}
where $\mathrm{MN}$ denotes a matrix normal distribution with row covariance $\I_n$ and column covariance $\Sigma$.

The Gaussian block model is a working model. In imaging applications it is typically applied after voxelwise standardization or another robust blockwise scaling step. Section~\ref{sec:practical} discusses practical checks and robustification.

\subsection{Prior and complete-data posterior}
We use the conjugate matrix-normal inverse-Wishart prior,
\begin{align}
\Sigma &\sim \mathrm{IW}(\nu_0,S_0),\label{eq:priorSigma}\\
B\mid \Sigma &\sim \mathrm{MN}_{k\times p}(B_0,V_0,\Sigma),\label{eq:priorB}
\end{align}
where $\nu_0>p+1$, $S_0\succ 0$, $B_0\in\R^{k\times p}$, and $V_0\succ0$. A convenient default in standardized blocks is $B_0=0$ and $V_0=(\alpha \I_k)^{-1}$ with $\alpha>0$.

For complete $Y$, define
\begin{equation}
V_n = (X^\top X + V_0^{-1})^{-1},\qquad
B_n = V_n(X^\top Y + V_0^{-1}B_0),
\label{eq:BnVn}
\end{equation}
and
\begin{equation}
S_n = S_0 + (Y - XB_n)^\top(Y - XB_n) + (B_n - B_0)^\top V_0^{-1}(B_n - B_0),\qquad
\nu_n = \nu_0 + n.
\label{eq:Sn}
\end{equation}

\begin{theorem}[Conjugate posterior under complete data]
\label{thm:conjugate_complete}
Under \eqref{eq:model}--\eqref{eq:priorB} with complete $Y$,
\begin{align}
\Sigma \mid Y,X &\sim \mathrm{IW}(\nu_n,S_n),\label{eq:sigma_post}\\
B\mid \Sigma,Y,X &\sim \mathrm{MN}_{k\times p}(B_n,V_n,\Sigma).
\label{eq:B_post}
\end{align}
\end{theorem}

For later use, the conditional posterior of $\Sigma$ given $B$ and complete $Y$ is inverse-Wishart,
\begin{equation}
\Sigma\mid B,Y,X \sim \mathrm{IW}\!\left(\nu_0+n+k,\,
\begin{aligned}
&S_0 + (Y-XB)^\top(Y-XB) \\
&\quad + (B-B_0)^\top V_0^{-1}(B-B_0)
\end{aligned}
\right).
\label{eq:sigma_post_given_B}
\end{equation}

\subsection{Conditional Gaussian draws for missing entries}
Fix $(B,\Sigma)$. For subject $i$, let $\mathcal{O}_i$ index observed entries and $\mathcal{M}_i$ index missing entries. Partition
\[
Y_i = \begin{pmatrix}Y_{i,\mathcal{O}_i}\\ Y_{i,\mathcal{M}_i}\end{pmatrix},\quad
\mu_i = X_iB = \begin{pmatrix}\mu_{i,\mathcal{O}_i}\\ \mu_{i,\mathcal{M}_i}\end{pmatrix},\quad
\Sigma = \begin{pmatrix}
\Sigma_{\mathcal{O}\mathcal{O}} & \Sigma_{\mathcal{O}\mathcal{M}}\\
\Sigma_{\mathcal{M}\mathcal{O}} & \Sigma_{\mathcal{M}\mathcal{M}}
\end{pmatrix}.
\]

\begin{lemma}[Gaussian conditional distribution]
\label{lem:cond_mvn}
Under \eqref{eq:model},
\begin{align}
Y_{i,\mathcal{M}_i}\mid Y_{i,\mathcal{O}_i},B,\Sigma,X_i
&\sim \mathcal{N}\Big(\mu_{i,\mathcal{M}_i} + \Sigma_{\mathcal{M}\mathcal{O}}\Sigma_{\mathcal{O}\mathcal{O}}^{-1}(Y_{i,\mathcal{O}_i}-\mu_{i,\mathcal{O}_i}),\nonumber\\
&\hspace{2.2cm}\Sigma_{\mathcal{M}\mathcal{M}} - \Sigma_{\mathcal{M}\mathcal{O}}\Sigma_{\mathcal{O}\mathcal{O}}^{-1}\Sigma_{\mathcal{O}\mathcal{M}}\Big).
\label{eq:cond_mvn}
\end{align}
\end{lemma}

When $\Sigma_{\mathcal{O}\mathcal{O}}$ is ill conditioned, direct inversion in \eqref{eq:cond_mvn} is unstable. HIMCE therefore supports stabilized conditioning. For $\delta>0$, define $\Sigma_{\delta}=\Sigma+\delta\I_p$ and apply Lemma~\ref{lem:cond_mvn} with $\Sigma_{\delta}$ in place of $\Sigma$. In standardized blocks, a convenient scale choice is $\delta = \varepsilon\,\tr(\Sigma)/p$ with $\varepsilon\in[10^{-6},10^{-2}]$.

\subsection{Covariance-mode updating}
\subsubsection{Inverse-Wishart posterior mode}
\begin{lemma}[Mode of inverse-Wishart]
\label{lem:iw_mode}
If $\Sigma\sim\mathrm{IW}(\nu,S)$ with $\nu>p+1$, then the unique mode is
\begin{equation}
\Sigma_{\mathrm{mode}} = \frac{S}{\nu+p+1}.
\label{eq:iw_mode}
\end{equation}
\end{lemma}

\subsubsection{Mode update conditional on $B$ and a completed block}
Motivated by \eqref{eq:sigma_post_given_B} and Lemma~\ref{lem:iw_mode}, the high-dimensional reference version of HIMCE replaces inverse-Wishart sampling by a deterministic covariance update. Given a completed $Y^\ast$ and current $B$, define residuals $E^\ast=Y^\ast-XB$ and the regularized residual sum of squares matrix
\[
\mathrm{RSS}(B)= (E^\ast)^\top E^\ast + (B-B_0)^\top V_0^{-1}(B-B_0).
\]
The conditional posterior mode of $\Sigma$ given $(B,Y^\ast)$ is
\begin{equation}
\Sigma \leftarrow \frac{S_0 + \mathrm{RSS}(B)}{\nu_0+n+k+p+1}.
\label{eq:sigma_mode_update}
\end{equation}
This update can be viewed as an ECM-type covariance step embedded in a stochastic imputation algorithm: missing data are drawn in the I-step, mean parameters are updated stochastically, and the covariance is updated to a stabilized mode.

HIMCE preserves the Gaussian likelihood and the conditional imputation law \eqref{eq:cond_mvn}: for every missingness pattern, missing entries are still drawn from the Gaussian conditional distribution implied by the current mean and covariance. It also preserves posterior uncertainty in the mean parameters through stochastic coefficient draws, either by the exact matrix-normal update \eqref{eq:mn_draw} when feasible or by scalable ridge/local-ridge approximations. What is approximated in the high-dimensional branch is the posterior law of $\Sigma$: instead of drawing a full inverse-Wishart covariance matrix, the algorithm uses a covariance mode, optionally expanded by a one-dimensional scalar bridge. Thus HIMCE is a posterior-motivated stochastic approximation rather than an exact joint-model sampler in large blocks. The low-dimensional branch restores the exact conditional inverse-Wishart covariance draw \eqref{eq:sigma_post_given_B} when that draw is computationally inexpensive.

\subsubsection{Regularization for high-dimensional blocks}
When $p>n$, the empirical residual covariance $n^{-1}(E^\ast)^\top E^\ast$ is singular. Even with a proper $S_0$, it can be beneficial to apply additional regularization before using \eqref{eq:sigma_mode_update}. A shrinkage option replaces $\widehat\Sigma_E=(E^\ast)^\top E^\ast/n$ by
\begin{equation}
\widehat\Sigma_{\mathrm{shrink}} = (1-\gamma)\widehat\Sigma_E + \gamma\,\frac{\tr(\widehat\Sigma_E)}{p}\I_p,\qquad \gamma\in[0,1],
\label{eq:shrink}
\end{equation}
and a ridge option uses
\begin{equation}
\widehat\Sigma_{\mathrm{ridge}} = \widehat\Sigma_E + \tau \I_p,\qquad \tau>0.
\label{eq:ridge_cov}
\end{equation}
The Ledoit--Wolf shrinkage choice is a standard data-driven option \citep{LedoitWolf2004}.

\subsubsection{Empirical-Bayes covariance approximation used by HIMA}
The numerical experiments benchmark HIMA as in Lu et al.~\citep{lu2025new}, following the empirical-Bayes covariance estimator implemented there and related to empirical-Bayes normal covariance shrinkage \citep{Champion2003}. Let $W=Y^\ast-XB$ denote the current residual block, and let $S_W$ and $R_W=(r_{ij})$ denote its pairwise sample covariance and correlation matrices. HIMA constructs a structured target matrix $Z$ with diagonal entries $Z_{ii}=S_{W,ii}$ and off-diagonal entries
\begin{equation}
Z_{ij}=\bar\rho\sqrt{Z_{ii}Z_{jj}},\qquad i\neq j,
\label{eq:z_cov_hima}
\end{equation}
where the common corrected correlation is
\begin{equation}
\bar\rho = \frac{1}{p(p-1)}\sum_{i\neq j} r_{ij}\,\Fhyper\!\left(\frac12,\frac12;\frac{n-1}{2};1-r_{ij}^2\right).
\label{eq:rho_bar_hima}
\end{equation}
Define
\begin{align}
\alpha_{ij} &= r_{ij}\,\Fhyper\!\left(\frac12,\frac12;\frac{n-1}{2};1-r_{ij}^2\right),\label{eq:alpha_hima}\\
\beta_{ij} &= 1-\frac{(n-2)(1-r_{ij}^2)}{n-1}\,\Fhyper\!\left(1,1;\frac{n+1}{2};1-r_{ij}^2\right),
\label{eq:beta_hima}
\end{align}
and estimate the shrinkage intensity through
\begin{equation}
k^2 = \frac{\sum_{i\neq j}\{\beta_{ij}-2\alpha_{ij}\bar\rho+\bar\rho^2\}}{\sum_{i\neq j}(1-\bar\rho^2)^2},\qquad \lambda_{\mathrm{EB}} = k^{-2}-3.
\label{eq:lambda_hima}
\end{equation}
The empirical-Bayes posterior mean and mode are then
\begin{equation}
\widehat\Sigma^{\mathrm{EB}}_{\mathrm{mean}} = \frac{\lambda_{\mathrm{EB}} Z + S_W}{\lambda_{\mathrm{EB}}+n},\qquad
\widehat\Sigma^{\mathrm{EB}}_{\mathrm{mode}} = \frac{\lambda_{\mathrm{EB}} Z + S_W}{\lambda_{\mathrm{EB}}+n+2p+2}.
\label{eq:eb_mode_hima}
\end{equation}
HIMA uses $\widehat\Sigma^{\mathrm{EB}}_{\mathrm{mode}}$ as its deterministic covariance update. In the companion R implementation, the two Gauss hypergeometric terms are evaluated by truncated series with a fixed number of expansion terms, and a nearest symmetric positive-definite projection is applied if roundoff makes the empirical-Bayes estimate indefinite.

\begin{remark}[Theory versus implementation]
The fixed-dimension theorems in Section~\ref{sec:theory} are proved for the conjugate inverse-Wishart mode plug-in \eqref{eq:sigma_mode_update}. The empirical-Bayes HIMA covariance mode \eqref{eq:eb_mode_hima} is the high-dimensional implementation used in the experiments. It instantiates the same covariance-mode plug-in principle, but it is not itself the object of those fixed-dimension proofs.
\end{remark}

When the block dimension is small, the computational reason for avoiding inverse-Wishart sampling disappears. The companion implementation therefore uses the exact conditional covariance draw \eqref{eq:sigma_post_given_B} in low-dimensional blocks, with $B_0=0$ and $V_0^{-1}=\alpha I_k$, so the scale matrix is $S_0 + (Y-XB)^\top(Y-XB) + \alpha B^\top B$. This branch is used in the repeated NHANES benchmark, where $p=2$, and directly addresses the under-coverage observed when the empirical-Bayes covariance mode is reused unchanged.

\subsection{Coefficient updates}
\subsubsection{Exact matrix-normal draw}
Under complete data and fixed $\Sigma$, the conditional posterior \eqref{eq:B_post} implies the exact draw
\begin{equation}
B \leftarrow B_n + L_V Z L_\Sigma^\top,
\qquad Z\in\R^{k\times p}\ \text{with i.i.d. }\mathcal{N}(0,1),
\label{eq:mn_draw}
\end{equation}
where $L_VL_V^\top=V_n$ and $L_\Sigma L_\Sigma^\top=\Sigma$. This update preserves cross-column dependence through $\Sigma$.

\subsubsection{Scalable ridge and local-ridge draws}
For large $p$, explicitly forming $L_\Sigma$ and the full matrix-normal draw \eqref{eq:mn_draw} can be expensive. HIMCE therefore uses a scalable approximation based on column-wise ridge posteriors. For column $j$,
\begin{equation}
Y_{\cdot j} \mid X,b_j,\sigma_j^2 \sim \mathcal{N}_n(Xb_j,\sigma_j^2 \I_n),\qquad
b_j \mid \sigma_j^2 \sim \mathcal{N}_k\!\left(0,\frac{\sigma_j^2}{\alpha}\I_k\right),
\label{eq:ridge_col}
\end{equation}
with $\sigma_j^2$ taken from the current covariance iterate, for example $\sigma_j^2=\Sigma_{jj}$. Then
\begin{equation}
 b_j \mid Y^\ast_{\cdot j},X,\sigma_j^2 \sim \mathcal{N}_k\left(\widehat b_j,\,\sigma_j^2 (X^\top X + \alpha \I_k)^{-1}\right),
\qquad
\widehat b_j = (X^\top X + \alpha \I_k)^{-1}X^\top Y^\ast_{\cdot j}.
\label{eq:ridge_posterior}
\end{equation}
In the primary high-dimensional companion benchmark, this global column-wise draw is replaced by a screened stochastic local-ridge draw that uses the forced design variables $X$ and a small set of highly correlated block predictors. This reduces global shrinkage of the posterior mean while keeping the mean update sparse and stable.

\subsection{Algorithms}
\subsubsection{Reference Bayesian sampler}
Algorithm~\ref{alg:exact} records the exact data augmentation sampler for the Gaussian block model. It targets the joint posterior of $(Y_{\mathrm{mis}},B,\Sigma)$ given $(Y_{\mathrm{obs}},X)$.

\begin{algorithm}[t]
\caption{Reference MVN data augmentation sampler}
\label{alg:exact}
\begin{algorithmic}[1]
\Procedure{MVN-DA}{$Y_{\mathrm{obs}},X,T$}
\State Initialize $(Y_{\mathrm{mis}}^{(0)},B^{(0)},\Sigma^{(0)})$.
\For{$t=1,\ldots,T$}
  \State Draw $Y_{\mathrm{mis}}^{(t)}$ from \eqref{eq:cond_mvn} given $(B^{(t-1)},\Sigma^{(t-1)})$.
  \State Draw $\Sigma^{(t)}$ from \eqref{eq:sigma_post}, then draw $B^{(t)}$ from \eqref{eq:B_post}.
\EndFor
\State \Return posterior draws $\{Y_{\mathrm{mis}}^{(t)},B^{(t)},\Sigma^{(t)}\}$.
\EndProcedure
\end{algorithmic}
\end{algorithm}

\subsubsection{HIMA and HIMCE covariance-mode chains}
Algorithm~\ref{alg:covmode} summarizes the covariance-mode chains used in the benchmark. HIMA uses deterministic ridge mean updates and the empirical-Bayes covariance mode \eqref{eq:eb_mode_hima}. HIMCE keeps the same conditional Gaussian imputation step but uses stochastic mean updates. In high-dimensional blocks, HIMCE uses the covariance mode \eqref{eq:eb_mode_hima}, optionally with a scalar bridge; in low-dimensional blocks, it uses the exact conditional covariance draw \eqref{eq:sigma_post_given_B}. In the theoretical reference chain, the high-dimensional covariance step may be represented by the conjugate mode update \eqref{eq:sigma_mode_update}.

\begin{algorithm}[t]
\caption{Covariance-mode chains for HIMA and HIMCE}
\label{alg:covmode}
\begin{algorithmic}[1]
\Procedure{CovModeMI}{$Y_{\mathrm{obs}},X,M,T,\delta,\alpha,\mathrm{method}$}
\State Initialize $Y^{(0)}$, $B^{(0)}$, and $\Sigma^{(0)}$ by a ridge fit and stabilized residual covariance.
\For{$m=1,\ldots,M$}
  \For{$t=1,\ldots,T$}
    \State Draw or fill $Y_{\mathrm{mis}}^{(t)}$ from \eqref{eq:cond_mvn} using $\Sigma_\delta$.
    \State Update mean: deterministic ridge mean for HIMA; stochastic ridge/local-ridge draw for HIMCE.
    \State Update covariance: empirical-Bayes mode \eqref{eq:eb_mode_hima} in high-dimensional blocks, or exact conditional inverse-Wishart draw \eqref{eq:sigma_post_given_B} for low-dimensional HIMCE.
    \State Optionally apply a scalar bridge to the high-dimensional HIMCE covariance iterate.
  \EndFor
  \State Store completed dataset $Y_{(m)}$.
\EndFor
\State For HIMCE, optionally apply the observed-cell calibration map to the stored draws.
\State \Return completed datasets $\{Y_{(m)}\}_{m=1}^M$.
\EndProcedure
\end{algorithmic}
\end{algorithm}

\subsubsection{Observed-cell calibration and reporting convention}
The companion implementation contains a finite-sample calibration layer after the core HIMCE chain. This layer is not part of the fixed-dimension theory. Let $\bar Y^{\mathrm{HIMA}}$ and $\bar Y^{\mathrm{HIMCE}}$ be posterior mean matrices computed from stored imputations. On a held-out set of observed cells, we fit
\[
\widetilde\mu_{ij}=a_j + b_j\bigl(w_j\bar Y^{\mathrm{HIMCE}}_{ij} + (1-w_j)\bar Y^{\mathrm{HIMA}}_{ij}\bigr),
\]
with $w_j$ shrunk toward HIMCE. The stored HIMCE deviations are recentered around $\widetilde\mu$. An additional scalar scale correction is accepted only if it improves observed-cell diagnostics computed from the same randomized empirical predictive CDF used for PIT, and only when the coverage gain is not offset by worse PIT or marginal-shape behavior. The calibration uses only observed cells; withheld pseudo-missing truths are used only for final evaluation.

In the empirical sections, \emph{HIMCE} denotes the final completed datasets after this observed-cell map. Pre-calibration stored draws are an internal stage of HIMCE rather than a separate method. HIMA and MICE are not passed through this calibration layer. The runtime column for HIMCE reports the shared empirical-Bayes fit plus the active chain; the calibration pass is exported separately in the companion output.

\section{Multiple-imputation inference and diagnostics}
\label{sec:diagnostics}
\subsection{Rubin pooling}
Let $Q$ be a scalar estimand, $\widehat Q_m$ its estimate from imputed dataset $m$, and $\widehat U_m$ the corresponding complete-data variance estimate. Rubin's combining rules are
\begin{align}
\bar Q &= \frac{1}{M}\sum_{m=1}^M \widehat Q_m,\qquad \bar U = \frac{1}{M}\sum_{m=1}^M \widehat U_m,\\
B_M &= \frac{1}{M-1}\sum_{m=1}^M (\widehat Q_m - \bar Q)^2,\qquad
T_M = \bar U + \left(1+\frac{1}{M}\right)B_M,
\end{align}
with degrees-of-freedom approximation
\[
\nu_{\mathrm{MI}} = (M-1)\left(1+\frac{1}{r}\right)^2,\qquad r = \frac{(1+1/M)B_M}{\bar U}.
\]
See \citep{Rubin2004}.

\subsection{Pseudo-missing calibration diagnostics}
Pointwise error summaries are informative for prediction but do not directly assess whether an imputation procedure produces an appropriate amount of variability. For continuous blocks it is often practical to perform pseudo-missing experiments starting from a fully observed subset.

Interquartile coverage is a direct check of dispersion. Starting from a complete block, mask entries completely at random at a fixed rate, rerun the imputer with many random seeds, and for each withheld cell compute the empirical first and third quartiles of the predictive distribution. The IQR coverage is the proportion of withheld truths that fall inside their central 50\% predictive intervals. Under rough calibration this value is close to $0.5$.

Support violations are informative when variables are bounded. For variables with natural bounds, record the fraction of imputed values that fall outside an admissible interval, such as the observed range in a benchmark subset. Large violations suggest that a Gaussian working model is too crude and motivate donor-based robustification such as predictive mean matching.

\subsection{Randomized rank-cell PIT and PIT-consistent empirical coverage}
\label{subsec:pit}
Coverage diagnostics are cell-wise, while many downstream analyses are sensitive to the aggregate distribution of imputed values. A complementary check compares the empirical distribution of withheld true values to the pooled empirical distribution of imputation draws.

Let $d_{\ell 1},\ldots,d_{\ell M}$ be the $M$ stored draws for withheld cell $\ell$, and let $t_\ell$ be the withheld truth. Define
\[
r^-_\ell = \#\{m: d_{\ell m}<t_\ell\},\qquad r^=_\ell = \#\{m: d_{\ell m}=t_\ell\}.
\]
The randomized rank-cell PIT is
\begin{equation}
\mathrm{PIT}_\ell = \frac{r^-_\ell + U_\ell(r^=_\ell+1)}{M+1},\qquad U_\ell\sim\mathrm{Unif}(0,1),
\label{eq:pit_rank_cell}
\end{equation}
independently across $\ell$. Relative to a left-edge empirical-rank PIT, \eqref{eq:pit_rank_cell} fills the whole rank cell and prevents finite-draw aliasing spikes from dominating the histogram.

For a well-calibrated posterior predictive distribution, the probability integral transform (PIT) values should be approximately uniform. A hump-shaped PIT distribution, with excess mass near the centre, indicates overdispersion: the predictive draws are too wide relative to the observed data. A U-shaped PIT distribution, with excess mass near 0 and 1, indicates underdispersion: the predictive draws are too narrow. Skewness in the PIT distribution indicates bias; excess mass near 0 suggests that predictions are systematically too large, whereas excess mass near 1 suggests that predictions are systematically too small.

The companion benchmark computes interval coverages from the same empirical predictive CDF used in \eqref{eq:pit_rank_cell}. For level $1-\alpha$, define
\begin{equation}
\widehat{\mathrm{cov}}_{1-\alpha} = \frac{1}{L}\sum_{\ell=1}^L \mathbf{1}\left\{\mathrm{PIT}_\ell\in\left[\frac{\alpha}{2},1-\frac{\alpha}{2}\right]\right\},
\label{eq:pit_consistent_cov}
\end{equation}
where $L$ is the number of withheld cells. Equation~\eqref{eq:pit_consistent_cov} yields PIT-consistent empirical central-interval coverage, so PIT and coverage cannot disagree merely because one is based on randomized empirical ranks and the other on interpolated type-7 quantiles from a small ensemble.

A convenient summary of central PIT mass is
\[
\widehat p_{0.4,0.6}=\frac{1}{L}\sum_{\ell=1}^L\mathbf{1}\{\mathrm{PIT}_\ell\in[0.4,0.6]\},
\]
whose uniform reference value is $0.20$.

\subsection{Marginal distribution overlays}
A second display compares the marginal distribution of withheld truths with the marginal distributions of pooled predictive draws and posterior means. In the updated companion benchmark, each panel overlays three densities: withheld truth, pooled draws, and posterior means. This allows center recovery and predictive spread to be assessed separately in a single figure. We also summarize marginal discrepancies by gaps in mean, standard deviation, interquartile range, and a quantile-quantile discrepancy statistic.

The distinction between posterior means and pooled draws is important. In our benchmark the block is standardized before masking, and HIMCE uses ridge-regularized mean updates. Both choices shrink extreme cell-wise posterior means toward the block center, so posterior means are naturally more concentrated than pooled predictive draws. A narrow posterior-mean density is therefore not, by itself, evidence that the imputation distribution is too narrow. Under-dispersion is diagnosed when concentrated posterior means are accompanied by narrow pooled-draw densities, U-shaped or edge-heavy randomized PIT, and low PIT-consistent coverage.

\section{Theoretical results}
\label{sec:theory}
This section records the key theoretical properties of the reference sampler and the covariance-mode approximation. Proofs are collected in \ref{app:proofs}.

\subsection{Correctness of the reference data augmentation kernel}
\begin{theorem}[Exactness of Algorithm~\ref{alg:exact}]
\label{thm:da_correct}
Assume \eqref{eq:model}, MAR \eqref{eq:mar}, and a proper prior \eqref{eq:priorSigma}--\eqref{eq:priorB}. The Markov kernel that alternates (i) sampling $Y_{\mathrm{mis}}$ from \eqref{eq:cond_mvn} and (ii) sampling $(B,\Sigma)$ from the full conditional posterior \eqref{eq:sigma_post} and \eqref{eq:B_post} leaves the joint posterior $p(Y_{\mathrm{mis}},B,\Sigma\mid Y_{\mathrm{obs}},X)$ invariant.
\end{theorem}

\subsection{Posterior consistency under complete data}
\begin{definition}[Frobenius neighborhoods]
For $\varepsilon>0$, define
\[
\mathcal{N}_\varepsilon = \{(B,\Sigma): \|B-B_\star\|_F + \|\Sigma-\Sigma_\star\|_F < \varepsilon\}.
\]
\end{definition}

\begin{theorem}[Posterior consistency under complete data]
\label{thm:consistency_complete}
Assume that $(Y_i,X_i)$ are i.i.d. with $Y_i\mid X_i\sim \mathcal{N}_p(X_iB_\star,\Sigma_\star)$, that $p$ and $k$ are fixed, and that $n^{-1}X^\top X\to Q_X$ with $Q_X\succ0$. Under a proper matrix-normal inverse-Wishart prior with positive density in a neighborhood of $(B_\star,\Sigma_\star)$,
\[
\Pi\big((B,\Sigma)\notin \mathcal{N}_\varepsilon \mid Y,X\big)\to 0
\]
in $P_{B_\star,\Sigma_\star}$-probability for every $\varepsilon>0$.
\end{theorem}

\subsection{Consistency of the covariance posterior mode}
\begin{theorem}[Posterior mode consistency for $\Sigma$]
\label{thm:mode_consistency}
Under the assumptions of Theorem~\ref{thm:consistency_complete}, let $\widehat\Sigma_{\mathrm{mode}}$ be the mode of the inverse-Wishart posterior \eqref{eq:sigma_post}. Then $\|\widehat\Sigma_{\mathrm{mode}}-\Sigma_\star\|_F\to 0$ in $P_{B_\star,\Sigma_\star}$-probability.
\end{theorem}

\subsection{Asymptotic validity of covariance-mode updating for prediction}
\begin{theorem}[Mode plug-in is asymptotically equivalent for prediction]
\label{thm:mode_plugin}
Under the assumptions of Theorem~\ref{thm:consistency_complete}, let $\Pi_n$ denote the exact posterior for $(B,\Sigma)$ under complete data. Let $\widehat\Sigma_{\mathrm{mode}}$ be the inverse-Wishart posterior mode. Consider the exact posterior predictive distribution
\[
\mathsf{P}_n(\cdot\mid Y,X)=\int p(\cdot\mid B,\Sigma,X)\,\Pi_n(dB,d\Sigma)
\]
and the plug-in predictive distribution
\[
\mathsf{P}_n^{\mathrm{mode}}(\cdot\mid Y,X)=\int p(\cdot\mid B,\widehat\Sigma_{\mathrm{mode}},X)\,\Pi_n(dB\mid \widehat\Sigma_{\mathrm{mode}}),
\]
where \(\Pi_n(dB\mid \widehat\Sigma_{\mathrm{mode}})\) denotes the matrix-normal law \eqref{eq:B_post} evaluated at \(\Sigma=\widehat\Sigma_{\mathrm{mode}}\). Then
\[
\TV\big(\mathsf{P}_n(\cdot\mid Y,X),\mathsf{P}_n^{\mathrm{mode}}(\cdot\mid Y,X)\big)\to 0
\]
in $P_{B_\star,\Sigma_\star}$-probability.
\end{theorem}

\subsection{High-dimensional stability under shrinkage}
\begin{proposition}[Eigenvalue bounds for shrinkage covariance]
\label{prop:eig_bounds}
Let $\widehat\Sigma_E$ be a symmetric positive semidefinite matrix and let $\widehat\Sigma_{\mathrm{shrink}}$ be defined by \eqref{eq:shrink} with $\gamma\in(0,1]$. Then
\[
\lambda_{\min}(\widehat\Sigma_{\mathrm{shrink}})\ge \gamma\,\frac{\tr(\widehat\Sigma_E)}{p},
\qquad
\lambda_{\max}(\widehat\Sigma_{\mathrm{shrink}})\le (1-\gamma)\lambda_{\max}(\widehat\Sigma_E)+\gamma\,\frac{\tr(\widehat\Sigma_E)}{p}.
\]
If $\tr(\widehat\Sigma_E)>0$, then $\widehat\Sigma_{\mathrm{shrink}}$ is strictly positive definite and its condition number is controlled by $\gamma$ and the spectrum of $\widehat\Sigma_E$.
\end{proposition}

\section{Simulation study}
\label{sec:sim}
\subsection{Design}
The updated companion benchmark uses a fully reproducible spatially correlated Gaussian block that mimics a small region-of-interest voxel grid. Variables are placed on a two-dimensional lattice and the residual covariance is generated from an exponential spatial kernel plus a nugget. The simulation uses $n=80$ subjects and a block size of $p=40$ after voxel subsampling and variance filtering. The design matrix contains an intercept and one age-like covariate; the first 10 variables have a stronger slope than the remainder. The completed block is standardized before masking. Pseudo-missing values are generated completely at random at the target rate, and diagnostics are evaluated on the withheld cells.

The revised HIMA and HIMCE runs share a deterministic empirical-Bayes covariance fit followed by a short chain with burn-in 8, thin 1, and 2 inner updates per kept draw. In the primary spatial benchmark, $p=40$ exceeds the exact-refresh threshold, so HIMCE stays on the empirical-Bayes covariance branch. The reported HIMCE results use screened stochastic local-ridge mean updates, a light scalar bridge with $\mathrm{df}=18$ and $\mathrm{bridge\_max}=1.6$, followed by observed-cell center and empirical-quantile scale calibration learned from stored draws.

\subsection{Methods compared and tuning}
We compare HIMA, HIMCE, and screened MICE. HIMA is the deterministic empirical-Bayes covariance baseline and uses the covariance mode \eqref{eq:eb_mode_hima}. HIMCE uses the same covariance backbone in the spatial simulation, replaces deterministic mean updates by screened stochastic local-ridge draws, keeps the light scalar bridge in the active chain, and applies the observed-cell center/scale calibration described above. MICE is implemented with screened Gaussian conditional regressions based on a sparse predictor matrix.

All methods use the same standardized block, auxiliary design matrix $X$, pseudo-missing mask, and $M=20$ stored completed datasets. Screened MICE uses Gaussian FCS with $X$ forced into each conditional model, at most 20 within-block predictors retained by largest absolute pairwise correlation on available cases, and 10 FCS iterations per imputation. No lasso or elastic-net penalty is added; stabilization comes from standardization, predictor screening, and the numerical singularity protection in the Gaussian FCS implementation. HIMA and HIMCE share the deterministic empirical-Bayes mode fit with 18 iterations. HIMCE then uses burn-in 8, thin 1, two inner updates per stored draw, scalar bridge degrees of freedom 18, and maximum bridge factor 1.6. Thus all methods receive matched input information, while conditional-model complexity is bounded explicitly for MICE.

\subsection{Metrics}
Let $\mathcal{M}$ index withheld cells. We report RMSE and MAE of the posterior mean across imputations, wall-clock runtime, PIT central mass $\widehat p_{0.4,0.6}$, central 50\% coverage, and PIT-consistent 90\% and 95\% empirical coverages. We also report PIT summaries $(\mathrm{pit\_mean},\mathrm{pit\_sd},\mathrm{pit\_ks})$ and marginal distribution gaps $(\mathrm{mean\_gap},\mathrm{sd\_gap},\mathrm{iqr\_gap},\mathrm{qq\_gap})$. For HIMCE, predictive metrics are computed after the observed-cell calibration step. The runtime column records the shared empirical-Bayes fit plus the active chain; the calibration pass is exported separately in the companion output.

\subsection{Representative results}
Tables~\ref{tab:sim_metrics} and \ref{tab:sim_pit} summarize the updated primary spatial benchmark as means and standard deviations across repeated pseudo-missing masks and random seeds. The main empirical message is straightforward. HIMCE improves substantially on HIMA in posterior-mean error and narrows the calibration gap appreciably, while preserving essentially the same active-chain runtime. Relative to screened MICE, HIMCE attains slightly smaller posterior-mean RMSE and MAE in this benchmark and runs in less than half the time, but its interval coverage remains lower.

\begin{table}[t]
\centering
\footnotesize
\setlength{\tabcolsep}{4pt}
\caption{Primary pseudo-missing simulation summary from the updated companion benchmark. Entries are mean (sd) across repeated pseudo-missing masks and random seeds. RMSE and MAE are computed from posterior means across imputations. HIMCE denotes the final stored draws after observed-cell center/scale calibration.}
\label{tab:sim_metrics}
\begin{adjustbox}{max width=\textwidth}
\begin{tabular}{lccccccc}
\toprule
Method & RMSE & MAE & Time (s) & $\widehat p_{0.4,0.6}$ & IQR cov. & cov$_{90}$ & cov$_{95}$\\
\midrule
HIMA       & 0.7433 (0.0288) & 0.5829 (0.0226) & 13.3790 (2.8849) & 0.0870 (0.0134) & 0.2406 (0.0207) & 0.5517 (0.0162) & 0.7671 (0.0176) \\
HIMCE & 0.6340 (0.0198) & 0.5004 (0.0158) & 13.6240 (3.0129) & 0.1438 (0.0198) & 0.3616 (0.0260) & 0.7555 (0.0282) & 0.8758 (0.0173) \\
MICE       & 0.6462 (0.0154) & 0.5124 (0.0129) & 31.6800 (5.7653) & 0.2267 (0.0148) & 0.5510 (0.0160) & 0.9277 (0.0130) & 0.9651 (0.0066) \\
\bottomrule
\end{tabular}
\end{adjustbox}
\end{table}

\begin{table}[t]
\centering
\footnotesize
\caption{Primary simulation PIT summary from the updated companion benchmark. Entries are mean (sd) across repeated pseudo-missing masks and random seeds. Uniform is the target.}
\label{tab:sim_pit}
\begin{adjustbox}{max width=0.82\textwidth}
\begin{tabular}{lccc}
\toprule
Method & pit\_mean & pit\_sd & pit\_ks\\
\midrule
HIMA       & 0.4930 (0.0181) & 0.3864 (0.0049) & 0.1984 (0.0150) \\
HIMCE & 0.4963 (0.0241) & 0.3388 (0.0086) & 0.1145 (0.0227) \\
MICE       & 0.4983 (0.0108) & 0.2718 (0.0045) & 0.0453 (0.0093) \\
\bottomrule
\end{tabular}
\end{adjustbox}
\end{table}

The simulation diagnostics should be read together rather than one column at a time. The error summary in Table~\ref{tab:sim_metrics} shows that HIMCE has the smallest posterior-mean RMSE and MAE, edging out screened MICE while remaining effectively tied with HIMA in active-chain runtime. The calibration summary is more nuanced. Relative to HIMA, HIMCE raises IQR coverage from 0.2406 to 0.3616, 90\% coverage from 0.5517 to 0.7555, and 95\% coverage from 0.7671 to 0.8758. That is a substantial practical gain, but it still leaves a calibration gap relative to screened MICE, whose corresponding values are 0.5510, 0.9277, and 0.9651. The PIT summary in Table~\ref{tab:sim_pit} tells the same story: HIMCE is much closer to uniformity than HIMA, but MICE remains the best-calibrated method overall.

Two diagnostic patterns are worth emphasizing. First, posterior means are naturally tighter than pooled draws because the simulation block is standardized and the mean step is ridge-regularized; a narrow posterior-mean density is therefore not, by itself, evidence of under-dispersion. Second, residual edge mass in PIT should be interpreted as genuine under-dispersion only when it persists after rank-cell randomization and lines up with low PIT-consistent coverage and a narrow pooled-draw density. On that combined reading, the updated HIMCE implementation behaves like a fast middle ground between HIMA and screened MICE: it recovers most of the point-prediction advantage of the covariance-mode chain while materially improving calibration, though not fully matching MICE.

\subsection{Posterior predictive distribution comparisons}
Figure~\ref{fig:sim_dist_compare} summarizes two complementary views. The upper row overlays the empirical density of withheld true values with the empirical densities of pooled imputation draws and posterior means for each method. The lower row reports randomized rank-cell PIT histograms.

\begin{figure*}[t]
\centering
\thispagestyle{empty}
\IfFileExists{figures/sim_dist_compare.pdf}{\includegraphics[width=0.88\textwidth]{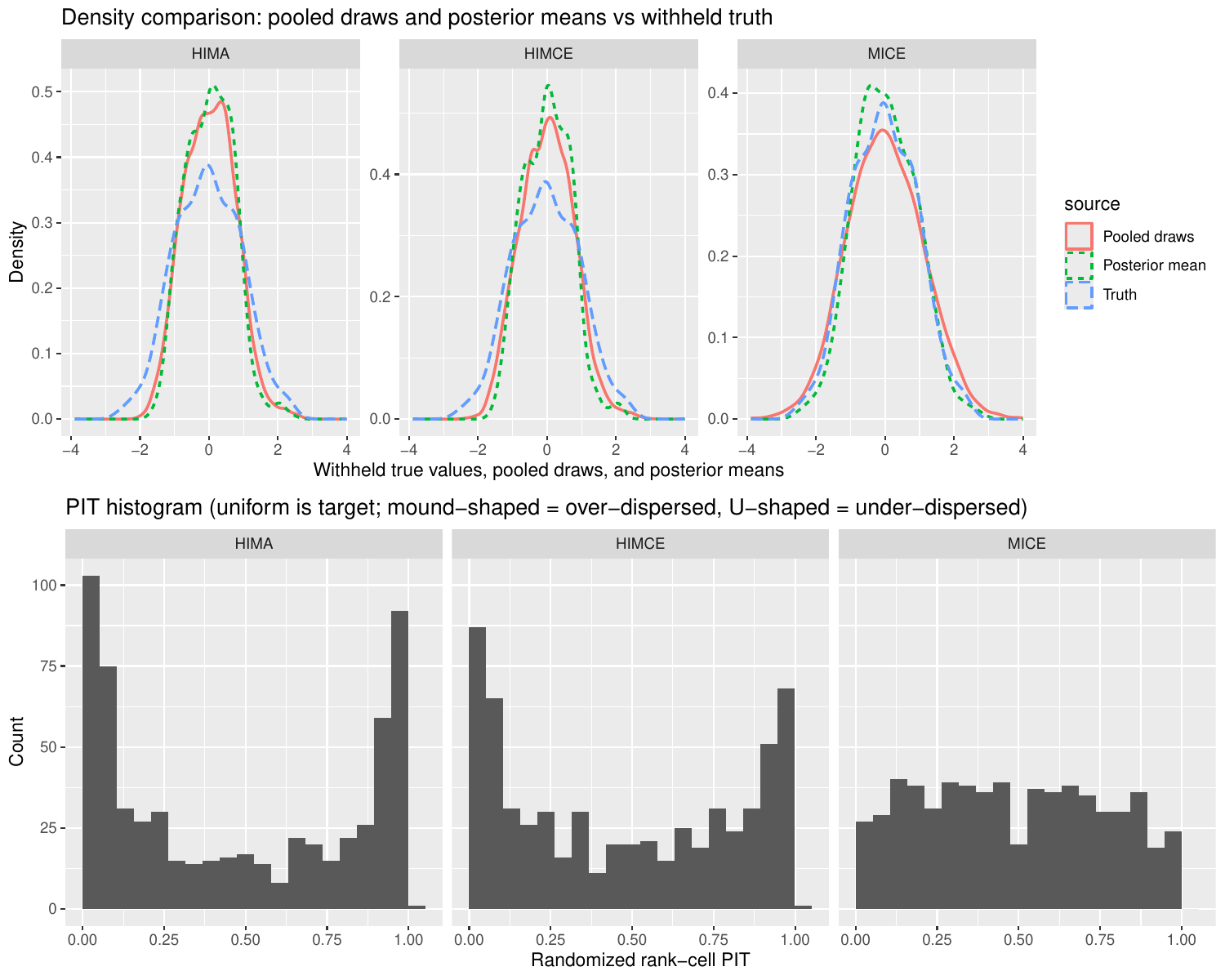}}{\IfFileExists{figures/sim_dist_compare.png}{\includegraphics[width=0.88\textwidth]{figures/sim_dist_compare.png}}{\fbox{\parbox{0.95\textwidth}{Run the companion R Markdown to generate \texttt{figures/sim\_dist\_compare.pdf}.}}}}
\caption{Distributional comparisons under pseudo-missing masking in the spatially correlated simulation. The upper row overlays withheld truths, pooled predictive draws, and posterior means. The lower row reports randomized rank-cell PIT histograms. Uniformity is the target. Excess central mass suggests over-dispersion, while excess mass near 0 and 1 suggests under-dispersion.}
\label{fig:sim_dist_compare}
\end{figure*}

\section{Real-data example: repeated NHANES pseudo-missing benchmark}
\label{sec:nhanes}
We retain a low-dimensional benchmark based on \texttt{nhanes2}, but in the revised companion it is evaluated through repeated pseudo-missing masks rather than a single small split. This change is important because only a small number of complete rows are available once the target block is restricted to \texttt{bmi} and \texttt{chl}. With such a small effective sample, one-shot PIT histograms and interval coverages are too discrete to support a stable ranking. Repeated masking therefore provides a more honest real-data illustration of the code path while keeping the example simple enough to inspect directly.

\subsection{Pseudo-missing setup}
We restrict to rows with complete outcomes, repeatedly mask a fraction of entries completely at random, apply each imputation method, and average the resulting pseudo-missing summaries. We take \texttt{bmi} and \texttt{chl} as the target block $Y$ and use age as the auxiliary design matrix $X$. HIMA uses the same empirical-Bayes covariance-mode update as in the spatial benchmark. HIMCE still uses the shared HIMA fit for initialization, but because the block has only $p=2$ variables it switches to the exact conditional inverse-Wishart covariance refresh \eqref{eq:sigma_post_given_B}. In this low-dimensional branch the default mean step is stochastic, scalar-bridge inflation is turned off because exact covariance uncertainty is already available, and the same observed-cell stored-draw calibration logic is retained. MICE remains the screened Gaussian FCS comparator.

\begin{table}[t]
\centering
\footnotesize
\setlength{\tabcolsep}{3.5pt}
\caption{Repeated NHANES pseudo-missing benchmark from the updated companion. Entries are mean (sd) across repeated pseudo-missing masks on the complete rows of \texttt{mice::nhanes2}. In this $p=2$ example, HIMCE uses the exact conditional covariance refresh rather than the high-dimensional covariance-mode plug-in.}
\label{tab:nhanes_metrics}
\begin{adjustbox}{max width=\textwidth}
\begin{tabular}{lccccccc}
\toprule
Method & RMSE & MAE & Time (s) & $\widehat p_{0.4,0.6}$ & IQR cov. & cov$_{90}$ & cov$_{95}$\\
\midrule
HIMA       & 0.9415 (0.3822) & 0.8117 (0.3315) & 0.1748 (0.0394) & 0.0213 (0.0649) & 0.0842 (0.1276) & 0.2229 (0.2057) & 0.5958 (0.1839) \\
HIMCE & 1.0244 (0.3323) & 0.9035 (0.3004) & 0.1565 (0.0396) & 0.1198 (0.1515) & 0.3172 (0.1998) & 0.7337 (0.2127) & 0.9115 (0.1222) \\
MICE       & 1.1177 (0.5610) & 0.9696 (0.4858) & 0.5713 (0.1050) & 0.1543 (0.1898) & 0.4225 (0.2757) & 0.8712 (0.2200) & 0.9343 (0.1378) \\
\bottomrule
\end{tabular}
\end{adjustbox}
\end{table}

The corresponding PIT summaries also favored HIMCE over HIMA: the mean PIT was 0.5009 for HIMCE versus 0.4842 for HIMA, and the PIT Kolmogorov--Smirnov statistic was 0.4094 for HIMCE versus 0.5699 for HIMA.

The NHANES illustration should be interpreted cautiously because it is based on only 13 complete rows, but it is still informative. The exact-refresh HIMCE branch raises 95\% coverage from 0.5958 under HIMA to 0.9115, while still improving on screened MICE in RMSE, MAE, and runtime. HIMA remains slightly better than HIMCE on point prediction in this tiny benchmark, but the coverage gain from restoring exact covariance uncertainty is substantial. We therefore treat the repeated NHANES results as a transfer-and-calibration illustration for the adaptive HIMCE design rather than as the main empirical claim.

\section{Practical considerations for neuroimaging and spatiotemporal arrays}
\label{sec:practical}
A practical operating mode for imaging is blockwise processing, such as ROI-wise or patch-wise operation. Each ROI or local patch is reshaped into an $n\times p$ block, and HIMA or HIMCE is run independently per block, optionally with shared subject-level covariates in $X$. This avoids a single massive covariance update across the entire image and aligns with common downstream ROI-wise analyses.

For spatiotemporal arrays, such as sensor-by-time or vertex-by-time blocks, block definitions should reflect the scientific target. One may treat each spatial location as a block of time points and include temporal basis functions in $X$, or impose separable spatial-temporal structure through a Kronecker covariance representation. The covariance-mode principle is compatible with either choice: the algorithm updates a stabilized covariance representation, but that representation may exploit spatial, temporal, or separable structure.

Initialization and tuning are important in high dimension. A stable default fills missing entries by design-aware column means or ridge predictions, initializes $B$ by ridge regression, and initializes $\Sigma$ by the HIMA empirical-Bayes covariance mode \eqref{eq:eb_mode_hima}. The ridge parameter $\alpha$ controls shrinkage in the mean update, and the hypergeometric truncation level controls the numerical approximation in \eqref{eq:rho_bar_hima}--\eqref{eq:lambda_hima}. In standardized blocks, $\alpha\in[1,100]$ and a truncation level around 25 terms are stable defaults in our experiments.

Diagnostics should be routine rather than optional. PIT is a diagnostic, not an imputation method, and PIT and coverage should be computed from the same empirical predictive CDF when only a modest number of imputations is stored. The observed-cell calibration layer is deliberately lightweight: it is learned from already stored draws, first corrects the center on observed cells, and then rescales draw deviations only when doing so improves PIT-consistent coverage without an unacceptable PIT or marginal-shape penalty.

The Gaussian block model is a working model. Bounded, skewed, heavy-tailed, or multimodal variables can show tail undercoverage, PIT edge mass, or support violations even when RMSE is acceptable. Standardization and monotone transformations can improve robustness; donor-based adjustments such as predictive mean matching are useful when support preservation is important. For heavier-tailed imaging or spatiotemporal blocks, robust covariance regularization or heavier-tailed residual models are natural extensions.

Computationally, the main gain comes from replacing repeated high-dimensional covariance sampling by stabilized mode updates and from exploiting shared linear algebra across columns. Memory pressure is usually driven by storing the completed block and the covariance iterate rather than by the ridge mean update itself, which is another reason ROI-wise and patch-wise processing are useful in imaging pipelines.

\section{Discussion}
\label{sec:discussion}
HIMCE targets a common bottleneck in neuroimaging and spatiotemporal analysis: multiple imputation for blocks with many correlated continuous variables. The method combines a regression mean model with a high-dimensional covariance plug-in, providing a practical alternative to high-dimensional FCS methods and to exact MVN data augmentation when repeated covariance sampling is fragile. In the experiments, that plug-in is the empirical-Bayes HIMA covariance estimator; in the theory, it is represented by a generic covariance-mode update with complete fixed-dimension proofs.

The updated benchmark clarifies the distinction between algorithm and calibration. In large blocks, HIMCE remains a covariance-mode method: the covariance step is the fast empirical-Bayes HIMA plug-in, and finite-sample calibration is applied afterward to stored draws. In small blocks, the main coverage repair comes from restoring exact covariance uncertainty through the inverse-Wishart conditional refresh. Empirically, HIMCE preserves the speed advantage of HIMA and improves point prediction and calibration, although screened MICE remains better calibrated in the spatial benchmark.

The benchmark also reinforces why diagnostic consistency matters. Randomized rank-cell PIT and PIT-consistent empirical-CDF coverages give a coherent view of calibration. Once that consistency is enforced, the earlier contradiction between ``good PIT'' and ``low nominal coverage'' largely disappears. This is especially important when the number of stored imputations is moderate, because naive interpolated quantiles can make the tails look too narrow purely as a finite-ensemble artifact.

More broadly, HIMCE should be interpreted as a posterior-motivated stochastic approximation to Gaussian joint-model imputation, not as an exact joint posterior sampler in its high-dimensional covariance-mode branch. It preserves the Gaussian conditional imputation structure and propagates mean-parameter uncertainty through stochastic coefficient or local-ridge draws. It approximates covariance uncertainty by replacing full matrix-valued inverse-Wishart sampling with covariance-mode updating, optionally augmented by a scalar bridge. The fixed-dimension theory justifies the generic mode plug-in principle asymptotically, but the empirical-Bayes HIMA covariance backbone used in the experiments is a high-dimensional implementation choice whose finite-sample adequacy remains an empirical question. Pseudo-missing calibration and marginal shape checks are therefore part of responsible use of the method, not optional embellishments.

Several extensions remain promising. One direction is to replace the scalar predictive bridge by a structured low-rank-plus-diagonal covariance expansion. Another is to generalize the observed-cell calibration layer to mixed data types while preserving the core covariance-mode chain. A third is to combine HIMCE with explicit spatial covariance models or separable spatiotemporal structure.

\section*{Funding}
This work was partially supported by the National Science Foundation under grants DMS-1924792 and DMS-2318925.

\section*{Code and data availability}
R code is provided with the submission to implement HIMCE and reproduce the numerical experiments reported in this paper. The code includes the HIMA and screened MICE comparators, the primary spatial pseudo-missing benchmark, and the repeated NHANES illustration. The diagnostic code computes randomized rank-cell PIT values and PIT-consistent central-interval coverages from the same empirical predictive CDF. The NHANES data used in the illustration are publicly available through the \texttt{mice} R package.

\section*{Declaration of competing interest}
The authors declare that they have no known competing financial interests or personal relationships that could have appeared to influence the work reported in this paper.

\section*{Declaration of generative AI use}
During the preparation of this work, the authors used AI-assisted tools to support language editing, organization, and manuscript revision. The authors reviewed and edited the content and take full responsibility for the final manuscript.

\appendix

\section{Proofs and additional results}
\label{app:proofs}

Throughout the appendix, \(\|\cdot\|_F\) denotes the Frobenius norm and \(\|\cdot\|_{\op}\) the operator norm. For a symmetric matrix \(A\), write \(\lambda_{\min}(A)\) and \(\lambda_{\max}(A)\) for its extreme eigenvalues. We use \(\mathrm{vec}(ABC)=(C^\top\otimes A)\mathrm{vec}(B)\). If \(Z\sim \mathrm{MN}_{a\times b}(0,U,V)\), then \(\mathrm{vec}(Z)\sim \mathcal{N}_{ab}(0,V\otimes U)\) and
\begin{equation}
\E\|Z\|_F^2=\tr(U)\tr(V).
\label{eq:mn_second_moment}
\end{equation}

\subsection{Proof of Lemma~\ref{lem:cond_mvn}}
Write \(Y_O=Y_{i,\mathcal{O}_i}\), \(Y_M=Y_{i,\mathcal{M}_i}\), and define \(\mu_O,\mu_M\) and the covariance blocks analogously. Under \eqref{eq:model},
\[
\begin{pmatrix}Y_O\\Y_M\end{pmatrix}
\sim \mathcal{N}\!\left(
\begin{pmatrix}\mu_O\\\mu_M\end{pmatrix},
\begin{pmatrix}\Sigma_{OO}&\Sigma_{OM}\\\Sigma_{MO}&\Sigma_{MM}\end{pmatrix}
\right).
\]
Since \(\Sigma\succ0\), \(\Sigma_{OO}\succ0\). Let
\[
S=\Sigma_{MM}-\Sigma_{MO}\Sigma_{OO}^{-1}\Sigma_{OM}
\]
be the Schur complement. The Schur complement is positive definite. The block inverse identity gives
\[
\Sigma^{-1}=\begin{pmatrix}
\Sigma_{OO}^{-1}+\Sigma_{OO}^{-1}\Sigma_{OM}S^{-1}\Sigma_{MO}\Sigma_{OO}^{-1} & -\Sigma_{OO}^{-1}\Sigma_{OM}S^{-1}\\
-S^{-1}\Sigma_{MO}\Sigma_{OO}^{-1} & S^{-1}
\end{pmatrix}.
\]
Substituting this inverse in the joint Gaussian density and retaining only terms depending on \(y_M\) gives
\[
-\frac12(y_M-\mu_M)^\top S^{-1}(y_M-\mu_M)
+(y_M-\mu_M)^\top S^{-1}\Sigma_{MO}\Sigma_{OO}^{-1}(y_O-\mu_O).
\]
Completing the square yields the kernel of a normal law with mean \(\mu_M+\Sigma_{MO}\Sigma_{OO}^{-1}(y_O-\mu_O)\) and covariance \(S\), which proves \eqref{eq:cond_mvn}. \qed

\subsection{Proof of Theorem~\ref{thm:conjugate_complete}}
The likelihood under \eqref{eq:matrixnormal} has kernel
\[
p(Y\mid B,\Sigma,X)\propto |\Sigma|^{-n/2}
\exp\left\{-\frac12\tr\big[\Sigma^{-1}(Y-XB)^\top(Y-XB)\big]\right\}.
\]
Multiplying by the matrix-normal and inverse-Wishart prior kernels gives
\begin{equation}
 p(B,\Sigma\mid Y,X)\propto
 |\Sigma|^{-(\nu_0+n+k+p+1)/2}
 \exp\left\{-\frac12\tr\big[\Sigma^{-1}A(B)\big]\right\},
\label{eq:app_joint_kernel}
\end{equation}
where
\[
A(B)=S_0+(Y-XB)^\top(Y-XB)+(B-B_0)^\top V_0^{-1}(B-B_0).
\]
Expanding the two quadratic terms in \(B\), collecting powers of \(B\), and using \(V_n=(X^\top X+V_0^{-1})^{-1}\) and \(B_n=V_n(X^\top Y+V_0^{-1}B_0)\) gives
\[
A(B)=S_n+(B-B_n)^\top V_n^{-1}(B-B_n),
\]
with \(S_n\) as in \eqref{eq:Sn}. Conditional on \(\Sigma\), \eqref{eq:app_joint_kernel} is therefore the kernel of \(\mathrm{MN}_{k\times p}(B_n,V_n,\Sigma)\), proving \eqref{eq:B_post}. The integral of this kernel over \(B\) is proportional to
\[
|V_n|^{p/2}|\Sigma|^{k/2}.
\]
Thus the factor \(|\Sigma|^{-k/2}\) in \eqref{eq:app_joint_kernel} is cancelled. The remaining kernel is \(\mathrm{IW}(\nu_0+n,S_n)\), proving \eqref{eq:sigma_post}. Holding \(B\) fixed gives \eqref{eq:sigma_post_given_B}. \qed

\subsection{Proof of Lemma~\ref{lem:iw_mode}}
For \(\Sigma\sim \mathrm{IW}(\nu,S)\), the log-density up to an additive constant is
\[
\ell(\Sigma)=-\frac{\nu+p+1}{2}\log|\Sigma|-\frac12\tr(S\Sigma^{-1}).
\]
Using \(d\log|\Sigma|=\tr(\Sigma^{-1}d\Sigma)\) and \(d\tr(S\Sigma^{-1})=-\tr(\Sigma^{-1}S\Sigma^{-1}d\Sigma)\), the first-order condition is
\[
(\nu+p+1)\Sigma^{-1}-\Sigma^{-1}S\Sigma^{-1}=0.
\]
Multiplication by \(\Sigma\) on the left gives \(S\Sigma^{-1}=(\nu+p+1)I_p\), so the only stationary point is \(\Sigma=S/(\nu+p+1)\). Since the inverse-Wishart density is continuous, vanishes on the boundary and at infinity in the positive definite cone, and has a unique stationary point, that point is the unique mode. \qed

\subsection{Proof of Theorem~\ref{thm:da_correct}}
Let
\[
\pi(y_{\mathrm{mis}},b,\sigma)=p(y_{\mathrm{mis}},b,\sigma\mid Y_{\mathrm{obs}},X)
\]
be the joint posterior. Algorithm~\ref{alg:exact} updates \(Y_{\mathrm{mis}}\), then \(\Sigma\), then \(B\), using the corresponding full conditional distributions under \(\pi\). For any bounded measurable test function \(h\), repeated conditioning gives
\begin{align*}
&\int \pi(dy_{\mathrm{mis}},db,d\sigma)
   \int \pi(dy'_{\mathrm{mis}}\mid b,\sigma)
   \int \pi(d\sigma'\mid y')
   \int \pi(db'\mid \sigma',y')h(y'_{\mathrm{mis}},b',\sigma') \\
&\qquad = \int h(y_{\mathrm{mis}},b,\sigma)\,\pi(dy_{\mathrm{mis}},db,d\sigma),
\end{align*}
where \(y'=(Y_{\mathrm{obs}},y'_{\mathrm{mis}})\). Thus \(\pi\) is invariant for the transition kernel. This is the standard Gibbs-sampler invariance argument applied to the three blocks \((Y_{\mathrm{mis}},\Sigma,B)\). \qed

\subsection{Auxiliary concentration and continuity facts}
\begin{lemma}[Inverse-Wishart concentration]
\label{lem:app_iw_conc}
Let \(\Sigma_\nu\sim \mathrm{IW}(\nu,S_\nu)\) with fixed dimension \(p\). If \(\nu\to\infty\) and \(S_\nu/\nu\to\Sigma_0\succ0\) in Frobenius norm, then \(\Sigma_\nu\to\Sigma_0\) in probability.
\end{lemma}
\noindent\textit{Proof.}
Use the representation \(\Sigma_\nu\stackrel{d}{=}S_\nu^{1/2}W_\nu^{-1}S_\nu^{1/2}\), where \(W_\nu\sim W_p(\nu,I_p)\). Since \(\nu^{-1}W_\nu\to I_p\) almost surely for fixed \(p\), we have \(\nu W_\nu^{-1}\to I_p\) almost surely. Combining this with \(S_\nu/\nu\to\Sigma_0\) and continuity of matrix multiplication proves the claim. \qed

\begin{lemma}[Total-variation continuity of Gaussian laws]
\label{lem:app_gaussian_tv}
Fix \(p\). On any set \(\{(\mu,\Sigma):mI_p\preceq \Sigma\preceq MI_p\}\) with \(0<m<M<\infty\), the map \((\mu,\Sigma)\mapsto \mathcal{N}_p(\mu,\Sigma)\) is continuous in total variation.
\end{lemma}
\noindent\textit{Proof.}
Pinsker's inequality and the Gaussian Kullback--Leibler formula give
\[
\TV(P,Q)^2\le \frac12\KL(P,Q),
\]
where, for \(P=\mathcal{N}_p(\mu,\Sigma)\) and \(Q=\mathcal{N}_p(\mu',\Sigma')\),
\[
\KL(P,Q)=\frac12\left\{\tr(\Sigma'^{-1}\Sigma)+
(\mu'-\mu)^\top\Sigma'^{-1}(\mu'-\mu)-p+
\log\frac{|\Sigma'|}{|\Sigma|}\right\}.
\]
The right side is continuous in \((\mu,\Sigma,\mu',\Sigma')\) on eigenvalue-bounded sets and equals zero when the two parameter pairs coincide. \qed

\subsection{Proof of Theorem~\ref{thm:consistency_complete}}
Write \(Y=XB_\star+E\), where the rows of \(E\) are i.i.d. \(\mathcal{N}_p(0,\Sigma_\star)\). From \eqref{eq:BnVn},
\begin{equation}
B_n-B_\star=V_nX^\top E+V_nV_0^{-1}(B_0-B_\star).
\label{eq:app_Bn_decomp}
\end{equation}
Since \(n^{-1}X^\top X\to Q_X\succ0\), \(\|V_n\|_{\op}=O(n^{-1})\). Also \(X^\top E=O_p(n^{1/2})\) in Frobenius norm, so the first term in \eqref{eq:app_Bn_decomp} is \(O_p(n^{-1/2})\), while the second is \(O(n^{-1})\). Hence \(B_n\to B_\star\) in probability.

Conditional on \(\Sigma\), \(B-B_n\sim \mathrm{MN}_{k\times p}(0,V_n,\Sigma)\). Therefore \eqref{eq:mn_second_moment} gives \(\E(\|B-B_n\|_F^2\mid \Sigma,Y,X)=\tr(V_n)\tr(\Sigma)\). Because \(\tr(V_n)=O(n^{-1})\) and \(\E\{\tr(\Sigma)\mid Y,X\}=\tr(S_n)/(\nu_n-p-1)=O_p(1)\), Markov's inequality implies
\[
\Pi(\|B-B_n\|_F>\varepsilon\mid Y,X)=o_p(1).
\]
Together with \(B_n\to B_\star\), this gives posterior concentration of \(B\) at \(B_\star\).

For \(\Sigma\), let \(R_n=Y-XB_n=E+X(B_\star-B_n)\). Expanding \(n^{-1}R_n^\top R_n\) gives
\[
\frac{1}{n}R_n^\top R_n=\frac{1}{n}E^\top E+o_p(1)=\Sigma_\star+o_p(1),
\]
where the last equality follows from the law of large numbers. The prior terms in \(S_n/n\) are negligible because \(S_0\) and \(V_0^{-1}\) are fixed and \(B_n=O_p(1)\). Thus \(S_n/n\to \Sigma_\star\). Since \(\nu_n/n\to1\), \(S_n/\nu_n\to\Sigma_\star\). Lemma~\ref{lem:app_iw_conc} applied conditionally to \(\Sigma\mid Y,X\sim\mathrm{IW}(\nu_n,S_n)\) gives posterior concentration of \(\Sigma\) at \(\Sigma_\star\). A union bound over the two concentration statements proves the theorem. \qed

\subsection{Proofs of Theorem~\ref{thm:mode_consistency} and Proposition~\ref{prop:eig_bounds}}
By Lemma~\ref{lem:iw_mode}, \(\widehat\Sigma_{\mathrm{mode}}=S_n/(\nu_n+p+1)\). The proof of Theorem~\ref{thm:consistency_complete} showed \(S_n/n\to\Sigma_\star\), while \((\nu_n+p+1)/n\to1\); hence \(\widehat\Sigma_{\mathrm{mode}}\to\Sigma_\star\), proving Theorem~\ref{thm:mode_consistency}.

For Proposition~\ref{prop:eig_bounds}, set \(\mu=\tr(\widehat\Sigma_E)/p\). For any unit vector \(v\),
\[
v^\top\widehat\Sigma_{\mathrm{shrink}}v=(1-\gamma)v^\top\widehat\Sigma_Ev+\gamma\mu.
\]
Since \(\widehat\Sigma_E\succeq0\), the lower bound follows. The upper bound follows from \(v^\top\widehat\Sigma_Ev\le\lambda_{\max}(\widehat\Sigma_E)\). If \(\mu>0\), the lower bound is strictly positive for all unit \(v\), so the matrix is positive definite. \qed

\subsection{Proof of Theorem~\ref{thm:mode_plugin}}
We prove the result for one fixed row with covariate vector \(x_i^\top\); any fixed finite collection of rows or missing entries is handled identically because \(p\) is fixed. Let \(P_\star=\mathcal{N}_p(x_i^\top B_\star,\Sigma_\star)\). By convexity of total variation in its first argument,
\[
\TV(\mathsf{P}_{n,i},P_\star)
\le \int \TV\{\mathcal{N}_p(x_i^\top B,\Sigma),P_\star\}\,\Pi_n(dB,d\Sigma).
\]
Theorem~\ref{thm:consistency_complete} places posterior mass in arbitrary small neighborhoods of \((B_\star,\Sigma_\star)\) with probability tending to one. On such neighborhoods, and for large \(n\), the covariance eigenvalues are bounded away from zero and infinity with high probability. Lemma~\ref{lem:app_gaussian_tv} then makes the integrand uniformly small; outside the neighborhood it is bounded by one. Thus \(\TV(\mathsf{P}_{n,i},P_\star)\to0\) in probability.

For the plug-in predictive law, Theorem~\ref{thm:mode_consistency} gives \(\widehat\Sigma_{\mathrm{mode}}\to\Sigma_\star\). Conditional on the plug-in covariance, the matrix-normal law \(\Pi_n(dB\mid\widehat\Sigma_{\mathrm{mode}})\) has mean \(B_n\) and row covariance \(V_n\), so the same second-moment argument used above gives concentration of \(B\) at \(B_\star\). Applying Lemma~\ref{lem:app_gaussian_tv} again yields \(\TV(\mathsf{P}^{\mathrm{mode}}_{n,i},P_\star)\to0\). The triangle inequality
\[
\TV(\mathsf{P}_{n,i},\mathsf{P}^{\mathrm{mode}}_{n,i})
\le \TV(\mathsf{P}_{n,i},P_\star)+\TV(\mathsf{P}^{\mathrm{mode}}_{n,i},P_\star)
\]
proves the claim. \qed

\section{Additional note on empirical-CDF diagnostics}
\label{app:ecdf}
The updated benchmark computes PIT and coverage from the same randomized empirical predictive CDF. This does not change HIMA or HIMCE; it only changes finite-ensemble calibration measurement. The identity behind \eqref{eq:pit_consistent_cov} is immediate: for any central level \(1-\alpha\), empirical coverage is exactly the fraction of PIT values lying in \([\alpha/2,1-\alpha/2]\). This prevents artificial disagreements between randomized empirical-rank PIT and interpolated quantile-based coverage in modest imputation ensembles.

\bibliographystyle{elsarticle-num}
\bibliography{references}

@book{Rubin2004,
  title        = {Multiple Imputation for Nonresponse in Surveys},
  author       = {Rubin, Donald B.},
  year         = {2004},
  publisher    = {Wiley},
  address      = {New York}
}

@article{vanBuuren2011,
  title   = {{mice}: Multivariate Imputation by Chained Equations in {R}},
  author  = {van Buuren, Stef and Groothuis-Oudshoorn, Karin},
  journal = {Journal of Statistical Software},
  year    = {2011},
  volume  = {45},
  number  = {3},
  pages   = {1--67}
}

@book{vanBuuren2018,
  title     = {Flexible Imputation of Missing Data},
  author    = {van Buuren, Stef},
  year      = {2018},
  edition   = {2},
  publisher = {Chapman and Hall/CRC},
  address   = {Boca Raton}
}

@book{LittleRubin2019,
  title     = {Statistical Analysis with Missing Data},
  author    = {Little, Roderick J. A. and Rubin, Donald B.},
  year      = {2019},
  edition   = {3},
  publisher = {Wiley},
  address   = {Hoboken}
}

@book{Schafer1997,
  title     = {Analysis of Incomplete Multivariate Data},
  author    = {Schafer, Joseph L.},
  year      = {1997},
  publisher = {Chapman and Hall/CRC},
  address   = {London}
}

@article{SchaferGraham2002,
  title   = {Missing Data: Our View of the State of the Art},
  author  = {Schafer, Joseph L. and Graham, John W.},
  journal = {Psychological Methods},
  year    = {2002},
  volume  = {7},
  number  = {2},
  pages   = {147--177}
}

@article{LedoitWolf2004,
  title   = {A Well-Conditioned Estimator for Large-Dimensional Covariance Matrices},
  author  = {Ledoit, Olivier and Wolf, Michael},
  journal = {Journal of Multivariate Analysis},
  year    = {2004},
  volume  = {88},
  number  = {2},
  pages   = {365--411}
}

@article{lu2025new,
  title     = {A New Multiple Imputation Method for High-Dimensional Neuroimaging Data},
  author    = {Lu, Tong and Kochunov, Peter and Chen, Chixiang and Huang, Hsin-Hsiung and Hong, L. Elliot and Chen, Shuo},
  journal   = {Human Brain Mapping},
  year      = {2025},
  volume    = {46},
  number    = {5},
  pages     = {e70161},
  publisher = {Wiley Online Library}
}

@article{Champion2003,
  title   = {Empirical {B}ayesian Estimation of Normal Variances and Covariances},
  author  = {Champion, Colin J.},
  journal = {Journal of Multivariate Analysis},
  year    = {2003},
  volume  = {87},
  number  = {1},
  pages   = {60--79},
  doi     = {10.1016/S0047-259X(02)00076-3}
}

\end{document}